\documentclass[10pt]{article}
\usepackage{amsfonts}
\usepackage{amsthm}
\usepackage[dvips]{graphicx,color}
\usepackage{graphicx}
\usepackage{xcolor}
\usepackage{amsmath}
\usepackage{indentfirst}
\usepackage[dvips]{graphicx,color}
\usepackage{authblk}

\newcommand{\fun}{{\cal F}}
\newcommand{\paral}{{\cal X}}
\newcommand{\killi}{{\cal Y}}

\newcommand{\qcd}{\begin{flushright} $\Box$ \end{flushright}}
\def\be{\begin{eqnarray}}
\def\ee{\end{eqnarray}}
\def\beq{\begin{equation}}
\def\eeq{\end{equation}}

\def\({\left (}
\def\){\right )}

\newtheorem{theorem}{Theorem}[section]
\newtheorem{lemma}[theorem]{Lemma}
\newtheorem{proposition}[theorem]{Proposition}
\newtheorem{corollary}[theorem]{Corollary}
\newtheorem{definition}{Definition}[section]
\newtheorem{remark}[theorem]{Remark}

\newtheorem{conjecture}[theorem]{Conjecture}
\title{Rigidity of geodesic completeness in the Brinkmann class of gravitational wave spacetimes}
\author[1]{I.P. Costa e Silva\footnote{Visiting Scholar. Permanent address: Department of Mathematics,
Universidade Federal de Santa Catarina, 88.040-900 Florian\'{o}polis-SC, Brazil.}}
\author[2]{J.L. Flores}
\author[3]{J. Herrera}
\affil[1]{\small{{\it Department of Mathematics,\\
University of Miami, Coral Gables, FL 33124, USA.}}}
\affil[2]{\small{{\it Departamento de \'Algebra, Geometr\'{i}a y Topolog\'{i}a,\\ Facultad de Ciencias, Universidad de M\'alaga \\ Campus Teatinos, 29071 M\'alaga, Spain.}}}
\affil[3]{\small{{\it Department of Mathematics,
Universidade Federal de Santa Catarina, 88.040-900 Florian\'{o}polis-SC, Brazil.}}}
\begin{document}

\maketitle

	\begin{abstract}
\noindent We consider restrictions placed by geodesic completeness on spacetimes possessing a null parallel vector field, the so-called {\em Brinkmann spacetimes}. This class of spacetimes includes important idealized gravitational wave models in General Relativity, namely the {\em plane-fronted waves with parallel rays}, or {\em pp-waves}, which in turn have been intensely and fruitfully studied in the mathematical and physical literatures for over half a century. More concretely, we prove a restricted version of a conjectural analogue for Brinkmann spacetimes of a rigidity result obtained by M.T. Anderson for stationary spacetimes. We also highlight its relation with a long-standing 1962 conjecture by Ehlers and Kundt. Indeed, it turns out that the subclass of Brinkmann spacetimes we consider in our main theorem is enough to settle an important special case of the Ehlers-Kundt conjecture in terms of the well known class of Cahen-Wallach spaces.

		
	\end{abstract}

\section{Introduction}

In 2000, M.T. Anderson proved a remarkable rigidity theorem \cite{anderson} establishing that {\em every geodesically complete, chronological, Ricci-flat $4$-dimensional stationary spacetime is isometric to (a quotient of) Minkowski spacetime}. (Recall that a spacetime, i.e., a connected time-oriented Lorentzian manifold, is said to be {\em stationary} if it admits a complete timelike Killing vector field.) The proof of this result is a powerful adaptation of the Cheeger-Gromov theory of sequences of collapsing Riemannian manifolds (see also \cite{cortier} for a much simpler proof of a slightly more restricted version of Anderson's theorem). The key importance of this work is not only due to the pioneering application of these techniques in General Relativity, but also to the striking insight it gives into the delicate nature of geodesic incompleteness in Lorentzian geometry.

Geodesic incompleteness of null or timelike geodesics has long been used in gravitational physics as a geometric description of gravitational collapse, such as that occurring at the center of black holes or at the Big Bang. Accordingly, that sort of geodesic incompleteness has indeed been rigorously shown to occur under reasonable, physically well-motivated conditions in the so-called {\em singularity theorems} of Mathematical Relativity \cite{BE,HE,oneill,wald}. Since gravity is thought to be always attractive (at least when one disregards quantum effects), gravitational systems are often unstable, and one expects gravitational collapse to be rather ubiquitous. This general idea led R. Geroch \cite{geroch} to conjecture that geodesically complete solutions to the Einstein field equations should be rare (in some suitable sense). Anderson's result can be viewed as one precise geometric realization of this physical idea.

Another such example appeared as early as 1962, in a separate development, when J. Ehlers and K. Kundt \cite[Section 2-5.7]{EK} put forth the conjecture that {\em every geodesically complete, Ricci-flat $4$-dimensional pp-wave is a plane wave}. A {\em (standard) pp-wave}\footnote{\label{foot1} Observe that pp-waves and plane waves can be defined intrinsically (i.e., in a coordinate-independent fashion - see, e.g., \cite[Definitions $1$ and $2$]{leistner}). Therefore, we use the term {\em standard} when there exists a preferred global coordinate system which allows one express the metric in a concrete way. Throughout the present paper, however, all pp-waves and plane waves considered will be standard in this sense, and so we shall omit this term unless there is risk of confusion. Nevertheless, this rule will not be applied to general Brinkmann spacetimes, because we will work simultaneously with general and standard ones.} is a spacetime of the form $(\mathbb{R}^{n}, g)$, where the metric $g$ is given in Cartesian coordinates $(u,v,x^1, \ldots, x^{n-2})$ by
\begin{equation}
\label{localBrink}
g = 2du(dv  + H(u,x)du) + \sum_{i=1}^{n-2}(dx^i)^2,
\end{equation}
and $H:\mathbb{R}^{n}\rightarrow \mathbb{R}$ is a smooth function independent of the $v$-coordinate. This class of spacetimes has been intensely studied both in the mathematical and physical literatures, since they give an idealized description of gravitational waves in General Relativity. A pp-wave where $H$ is quadratic on the $x$-coordinates, i.e., where
\[
H(u,x)=\sum_{i,j=1}^{n-2}a_{ij}(u)x^ix^j,
\]
is called a {\em (standard) plane wave}. Plane waves have a number of important physical and geometrical properties (see, e.g., Ch. 13 of \cite{BE} for some of these). The original version of the Ehlers-Kundt conjecture is still open, but there has been progress in obtaining partial results \cite{B,FlS,HR,leistner}.

Now, pp-waves are not stationary in general, but one can easily check from (\ref{localBrink}) that the vector field $\partial_v$ is {\em null} and {\em parallel} (i.e., with $\nabla \partial_v=0$; here and hereafter $\nabla$ will denote the Levi-Civita connection of the underlying metric tensor). Therefore, pp-waves are distinguished representatives of the larger class of {\em Brinkmann spaces}, that is, Lorentzian manifolds admitting a null parallel vector field $\paral$. Such manifolds owe their name to H.W. Brinkmann, who discovered them in 1925 \cite{brinkmann}. Brinkmann spaces have special Lorentzian holonomy, which in turn gives rise to a number of interesting geometric properties (see, e.g., \cite{batat,baum,galaev,leistner} and references therein for recent results). Brinkmann spaces are always time-orientable \cite{baum}, so there is no loss of generality in considering only connected, time-oriented Brinkmann spaces, which we will call simply Brinkmann spacetimes.

A nice way of viewing Brinkmann spacetimes is as null analogues of stationary spacetimes. One is then naturally led to consider the following null version of the Anderson's rigidity theorem, firstly conjectured in \cite{CF}.

\begin{conjecture}
\label{key}
Every strongly causal, Ricci-flat $4$-dimensional Brinkmann spacetime satisfying certain completeness condition is isometric to (a quotient of) a plane wave spacetime.
\end{conjecture}
Here, the phrase ``certain completeness condition'' parallels the condition, present in Anderson's theorem, that spacetime be geodesically complete. It is not clear to the authors if geodesic completeness alone would suffice in this context. Nevertheless, we do believe that geodesic completeness plus {\em transversal completeness} (see the definition just above Theorem \ref{thm:teoIJ}) should be enough to get the conclusion of Conjecture \ref{key}, since Theorems \ref{thm:teoIJ} and \ref{thm:teo1} show that, under these hypotheses, Conjecture \ref{key} reduces to the Ehlers-Kundt conjecture. On the other hand, the assumption that spacetime is {\em strongly causal}, i.e., has no ``almost closed'' nonspacelike curves, replaces the condition in Anderson's theorem that spacetime be {\em chronological}, that is, the absence of closed timelike curves \cite{CF}. These causal conditions imply that the Killing vector field in each case will give rise to an isometric action without fixed points, which is in turn convenient in taking quotients. While the exact analogue of chronology in the null case would be to require that spacetime be only {\em causal}, i.e., has no closed causal curves, it turns out, after a closer inspection, that a little more causality is required to have the mentioned quotient behave well \cite{CF}. Since in particular every plane wave spacetime is causally continuous \cite{EE}, and hence strongly causal, this assumption does not seem very restrictive.

Our main goal in this paper is to present a proof of the following restricted version of Conjecture \ref{key}.
\begin{theorem}
\label{main}
Let $(M,g)$ be a geodesically complete, strongly causal, Ricci-flat $4$-dimensional Brinkmann spacetime. If $(M,g)$ is transversally Killing, then the universal covering $(\overline{M}, \overline{g})$ of $(M,g)$ is isometric to a plane wave.
\end{theorem}

If $\paral \in \Gamma(TM)$ denotes the null parallel vector field in the Brinkmann spacetime $(M,g)$, the extra condition ``transversally Killing'' means, by definition, that there exists a Killing vector field $\killi \in \Gamma(TM)$ such that $g(\paral, \killi) =1$ (see discussion after Theorem \ref{thm:teoIJ} below). 
A concrete situation where this occurs is when the Brinkmann spacetime is an {\em autonomous} pp-wave (\ref{localBrink}), i.e., $H$ does not depend on the variable $u$. In this case $\killi := \partial_u$ will ensure that the pp-wave is indeed transversally Killing (with $\paral=\partial_v$). Therefore, from the proof of Theorem \ref{main} we deduce the following version of the Ehlers-Kundt conjecture.

\begin{corollary}
\label{keycor}
Every geodesically complete, {\em strongly causal, autonomous,} Ricci-flat, $4$-dimensional pp-wave is a Cahen-Wallach space.
\end{corollary}

Recall that a {\em Cahen-Wallach space} is an indecomposable, solvable geodesically complete symmetric Lorentzian manifold. These were classified by M. Cahen and N. Wallach in 1970 \cite{CW}, who showed the universal covering of any connected component of such a manifold is isometric to $(\mathbb{R}^{n}, g_{\lambda})$, where $\lambda := (\lambda_1, \ldots, \lambda _{n-2}) \in \mathbb{R} \setminus \{0\}$ and \[
g_{\lambda} = 2dudv + \sum_{i=1}^{n-2} \lambda _i (x^i)^2 du + \sum_{i=1}^{n-2}(dx^i)^2,
\]
which we recognize to be a plane wave without $u$-dependence.

There are several results related to Conjecture \ref{key} in the literature. For example, Leistner and Schliebner \cite{leistner} have recently shown that {\em the universal covering of any {\em compact} Ricci-flat Brinkmann spacetime is a plane wave}. Their result holds in any dimension, and is geometrically quite interesting, but it is unclear to us how strong the assumption of compactness actually is for this class of spacetimes. For instance, for Brinkmann spacetimes it in particular implies geodesic completeness \cite{leistner}, which unlike the Riemannian case does {\em not} follow automatically from compactness alone. At any rate, compact spacetimes are never even chronological \cite{oneill}, and therefore have arguably less physical interest. This, in part, has motivated our search for analogous results in the non-compact setting.

After finishing this paper, we became aware of another, much more general context in which the assumptions in Theorem \ref{main} are natural\footnote{We are grateful to J.M.M. Senovilla for bringing this to our attention, and for pointing out Refs. \cite{BSS,MS} to us.}, namely in a local classification scheme, recently carried out by M. Mars and J.M.M. Senovilla \cite{MS} (see also \cite{BSS}), for a class of algebraically special spacetimes which includes both Kerr and Brinkmann spacetimes (as well as generalizations of these). More specifically, in \cite{MS} the authors investigate $4$-dimensional Einstein spacetimes $(M,g)$ endowed with a Killing vector field $\killi \in \Gamma(TM)$ and satisfying a certain ``alignment'' relation between the Weyl tensor and the curl of $\killi$. In an important special case (cf. Theorem 1 of \cite{MS}) of their classification scheme as applied to Ricci-flat spacetimes, these authors show the global existence of a parallel null vector field $\paral \in \Gamma(TM)$ for which $g(\paral,\killi) =1$, and therefore end up precisely with what we call here a transversally Killing Brinkmann spacetime. Indeed, they show that these are {\em locally} isometric to autonomous pp-waves. (Mars and Senovilla call these {\em stationary vacuum Brinkmann spacetimes}.) Theorem \ref{main} can thus be viewed as a global rigidity result pertaining to such a subclass of spacetimes. 

The paper is organized as follows. In Section \ref{sec:prel}, we introduce the so-called {\em standard} Brinkmann spacetimes, and establish some of the terminology which we will need in the main proof. Brinkmann already knew that a 4-dimensional Ricci-flat Brinkmann spacetime is locally a pp-wave \cite{brinkmann} (see also \cite{galaev,schimming}), but we wish to globalize this result here, which we do in Section \ref{sec2}. After some technical lemmas, Theorem \ref{main} is proved in Section \ref{sec3}. We finish with a discussion at the end of this same section of a context in which our theorem implies that our spacetime is of Cahen-Wallach type.

\section{Preliminaries on Brinkmann spacetimes}
\label{sec:prel}
Let $(M^n,g)$ ($n\geq 3$) be a Brinkmann spacetime, which, recall, is a smooth\footnote{Here and hereafter, by {\em smooth} we always mean $C^{\infty}$.} connected time-oriented Lorentzian manifold admitting a complete null parallel vector field $\paral$ (i.e., with $\nabla \paral=0$). We will say that the Brinkmann spacetime $(M,g)$ is {\em standard} if $M=\mathbb{R}^2\times Q$ for some $(n-2)$-dimensional smooth manifold $Q$ which we shall call the {\em spatial fiber}, and the metric $g$ can be expressed as
%

\begin{equation}\label{SBM}
g=du\otimes\left(dv+H\,du+\Omega\right) + \left(dv+H\,du+\Omega\right)\otimes du + \gamma ,
\end{equation}
where:

\begin{itemize}
	\item[(a)] $\gamma$ is a smooth $(0,2)$-tensor on $\mathbb{R}^2\times Q$ whose radical at each $p=(v_0,u_0,x_0)\in \mathbb{R}^2\times Q$ is $span\{\partial_v|_{p},\partial_u|_{p}\}$, and so $Q\ni x\mapsto \gamma_{(v_0,u_0,x)}|_{T_{x}Q\times T_{x}Q}$ defines a smooth Riemannian metric on $Q$.
	
	\item[(b)] $\Omega$ is a smooth 1-form on $\mathbb{R}^2\times Q$ with $\Omega(\partial_v)=\Omega(\partial_{u})=0$, and so $Q\ni x\mapsto \Omega_{(v_0,u_0,x)}|_{T_{x}Q}$ defines a smooth 1-form on $Q$ and
	
	\item[(c)] $\gamma$, $\Omega$ and $H$ have no dependence on the $v$-coordinate as $\partial_v=\paral$ is in particular a Killing vector field.
\end{itemize}

As discussed in the Introduction, if $\Omega=0$ and $\gamma$ is the flat Euclidean metric on $Q=\mathbb{R}^2$, we will say that $(M,g)$ is a {\em (standard) pp-wave}. Moreover, a standard pp-wave where $H$ is quadratic on the $x$-coordinates will be called a {\em (standard) plane wave}. Finally, a standard Brinkmann spacetime will be called {\em autonomous} if the quantities $H,\Omega$ and $\gamma$ in \eqref{SBM} have no dependence on the coordinate $u$. Therefore, in the autonomous standard Brinkmann case the vector field $\partial_u$ is also a Killing vector field.

In general, a Brinkmann spacetime need not be standard. However, as recently shown by two of us (IPCS and JLF) \cite{CF}, it is possible to obtain mild conditions ensuring that a Brinkmann spacetime $(M,g)$ the complete parallel vector field $\paral$ can indeed be expressed in the standard form. This will happen, for instance, if $(M,g)$ is {\em transversally complete}, which means by definition that there exists a complete field $\killi \in \Gamma (TM)$ {\em conjugate} to $\paral$, in the sense that $g(\paral,\killi)=1$ and $[\paral,\killi]=0$. Concretely (see \cite[Theorem V.11]{CF}),

\begin{theorem}\label{thm:teoIJ}
	Let $(M,g)$ be a strongly causal and transversally complete Brink\-mann spacetime. Then, the universal covering $(\overline{M},\overline{g})$ of $(M,g)$ is isometric to a standard Brinkmann spacetime (\ref{SBM}). The isometry can be chosen to be such that it associates the lift $\overline{\paral}$ of $\paral$ to $\partial_v$ and the lift $\overline{\killi}$ of $\killi$ to $\partial_u$.
\end{theorem}
We will say that a Brinkmann spacetime $(M,g)$ is {\em transversally Killing} if there exists a (not necessarily complete) Killing field $\killi$ conjugate to the complete parallel vector field $\paral$ of $(M,g)$. Actually, as long as $\killi$ is Killing, we need only impose {\em either} that $ [\paral,\killi] =0$ on $M$ and $g(\paral (p), \killi(p)) =1 $ at a single point $p \in M$, {\em or} that $g( \paral , \killi ) =1 $ on $M$ in order to ensure that $(M,g)$ is transversally Killing:

\begin{proposition}
\label{prop:killing}
Let $(M,g)$ be a spacetime with a parallel vector field $\paral$ and a Killing vector field $\killi$. Then $ [\paral,\killi] =0$ if and only if  $g( \paral , \killi ) $ is constant throughout $M$.
\end{proposition}
{\em Proof.} Given any vector field $Z\in \Gamma (TM)$, we have
\[
d(g( \paral ,\killi ))(Z) = Z g( \paral ,\killi ) = g( \paral ,\nabla_Z\killi ) = - g( Z ,\nabla _{\paral} \killi ) = g( Z ,[\killi, \paral]),
\]
where we have used that $\killi$ is Killing on the third equality and that $\paral$ is parallel on the second and the last equalities. 
\qcd

As a consequence of Theorem \ref{thm:teoIJ}, we have the following.

\begin{corollary}\label{thm:teoIJ2}
Let $(M,g)$ be a strongly causal, geodesically complete and transversally Killing Brink\-mann spacetime. Then, the universal covering $(\overline{M},\overline{g})$ of $(M,g)$ is isometric to a standard autonomous Brinkmann spacetime.
\end{corollary}
\begin{proof}
	Note that since $(M,g)$ is transversally Killing, there exists a Killing vector field $\killi$ such that $g(\paral,\killi)=1$ and $[\paral,\killi]=0$. Since $(M,g)$ is geodesically complete, $\killi$ is actually a complete vector field conjugate to $\paral$, and so $(M,g)$ is also transversally complete. Thus, we can apply Theorem \ref{thm:teoIJ} and obtain that the universal cover of $(M,g)$ is isometric to a standard Brinkmann spacetime. Moreover, the isometry can be chosen to be such that $\partial_u$ is associated to $\overline{\killi}$ the lift of $\killi$, and so $\partial_u$ is Killing. In particular, the standard Brinkmann spacetime is autonomous.
\end{proof}

Without any causality assumptions we have the following rigidity result, whose proof uses (a version of) Theorem 3 of \cite{leistner}. 

\begin{proposition}
\label{prop:nocausal}
Let $(M,g)$ be a geodesically complete, Ricci-flat $4$-dimensional Brinkmann spacetime. If $(M,g)$ is transversally Killing, then the universal covering $(\overline{M}, \overline{g})$ of $(M,g)$ is isometric to a standard pp-wave.
\end{proposition}
\begin{proof} Since $(M,g)$ is Ricci-flat and $4$-dimensional, it is locally a standard pp-wave (see, e.g., \cite{galaev} or the proof of Theorem \ref{thm:teo1} below), and therefore we have \cite{leistner} that \[
R(V,W) =0 , \forall V,W \in \paral ^{\perp}.
\]
But then, since there exists a (complete) Killing vector field $\killi$ with $g(\paral,\killi)=1$, $(M,g)$ is a pp-wave (in the intrinsic way, see Footnote \ref{foot1}) and  Theorem 3 of \cite{leistner} yields\footnote{In \cite{leistner} the authors assume in their Theorem 3 the existence of a complete {\em null} vector field $Y$ such that $g( \paral, \killi) =1$, but actually the causal character of $\killi$ is not used anywhere in the proof.} the result. \end{proof}
\smallskip

\begin{remark}
\label{rem:leist} {\em 
We shall need to give below an alternative proof of Proposition \ref{prop:nocausal} which uses strong causality. The justification is that this proof has the advantage of giving a very concrete form for the effect of the Killing vector field on the function $H$ (cf. Remark \ref{rem:aux} below). This will be crucial for our main proof. Moreover, strong causality is needed elsewhere in the proof anyway, so there is no real loss of generality in that causal assumption.}
\end{remark}

We end this section with some comments regarding notation. A coordinate system on a standard Brinkmann spacetime will be often denoted by $\{u,v,x^1,\dots,x^{n-2}\}$, where $\{x^1,\dots,x^{n-2}\}$ is a local coordinate system for $Q$. We will denote generic spatio-temporal indices by greek letters $\alpha,\beta, \dots$, and indices on the spatial fiber by latin letters $i,j,\dots$. We will also make use of $u,v$ for the corresponding indices, to avoid confusion with spatial fiber indices. We use throughout the Einstein's summation convention. Finally, the superscript ``$Q$'' indicates (covariant or exterior) differentiation and/or geometric quantities on the Riemannian manifold $(Q,\gamma)$.

 \section{From a standard Brinkmann spacetime to a pp-wave}\label{sec2}
Our first aim in this paper is to show that under conditions analogous to those appearing in Anderson's rigidity theorem, standard Brinkmann spacetimes are isometric to pp-waves. Concretely,

\begin{theorem}\label{thm:teo1}
	{\em Let $(M^n,g)$ be a standard Brinkmann spacetime and assume that:
		\begin{itemize}
			\item[i)] $M$ is simply connected,
			\item[ii)] $n=4$,
			\item[iii)] $(M,g)$ is geodesically complete, and
			\item[iv)] $Ric =0$, i.e., $(M,g)$ is Ricci-flat.
		\end{itemize}
		Then $(M,g)$ is isometric to a $pp$-wave. In fact, $M = \mathbb{R}^4$ and there exist coordinates $\{U,V,X,Y\}$ with $(U,V,X,Y)\in \mathbb{R}^4$, such that
		
		\[g = 2 dU (dV + \tilde{H}(U,X,Y)dU) + dX^2 + dY^2,\quad\hbox{with $\tilde{H}$ harmonic in $X,Y$.}
\]
				}
\end{theorem}
\begin{proof} Note, first of all, that (i) implies that $Q$ is connected and simply connected. Pick local coordinates $\{x^1,\dots,x^{n-2}\}$ on an open connected and simply connected patch $U \subseteq Q$, together with the given global coordinates $(u,v)$ on the $\mathbb{R}^2$ part. On the neighborhood covered by the coordinates $u,v,x^1,\dots,x^{n-2}$, the metric \eqref{SBM} becomes
\begin{equation}
\label{ref2}
g = 2 du (dv + H(u,x)du + \Omega_i(u,x)dx^i) + \gamma_{ij}(u,x)dx^i dx^j,
\end{equation}
where $x = (x^1,\dots,x^{n-2})$. A direct computation of the Christoffel symbols shows that the only non-zero ones are
\begin{equation}\label{chrissym}
 \begin{array}{ll}
 \Gamma ^i_{uu} = - (\nabla^Q H)^i +  \gamma ^{ij}\frac{\partial \Omega _{j}}{\partial u}, & \Gamma ^i _{uk} = \Gamma ^i _{ku}  = -\gamma ^{ij}(d^Q \Omega)_{jk} + \gamma ^{ij}\frac{\partial \gamma  _{jk}}{\partial u}, \\
 \Gamma ^i_{jk} = (\Gamma^Q)^i_{jk},  & \Gamma ^v _{ku}  = \frac{\partial H}{\partial x^k} + \gamma^{ij}\Omega _i (d^Q \Omega)_{jk} -  \gamma ^{ij}\Omega _i\frac{\partial \gamma _{jk}}{\partial u}, \\
 \Gamma ^v _{ij} = \frac{1}{2}[(\nabla^Q)_i \Omega _j + (\nabla^Q)_j \Omega _i -  \frac{\partial \gamma_{ij}}{\partial u}]\quad & \Gamma ^v _{uu}  = \frac{\partial H}{\partial u} + \Omega _i (\nabla^Q H)^i -  \gamma^{ij}\Omega _i\frac{\partial \Omega _{j}}{\partial u}.
 \end{array}
 \end{equation}
Again, a direct calculation reveals that
\begin{equation}
\begin{array}{ccc}
R^i_{jkl} = (R^Q)^i_{jkl}, & R^v_{\alpha v \beta} =0, & R^u_{j u l} =0,
\end{array}
\end{equation}
so that $0 = (Ric)_{jl} = (Ric^Q)_{jl}$, i.e., $(Q,\gamma)$ is Ricci-flat. Specializing to $n=4$, we have that dim $Q$ =2, so $Q$ is flat. Since $Q$ is simply connected and bidimensional, it is diffeomorphic to $\mathbb{R}^2$, and we may take $U \equiv Q$, which we do from now on. We may, therefore, select {\em a posteriori} coordinates $x:=x^1$ and $y:= x^2$ so that $\gamma_{ij} = \delta_{ij}$, and thus $\Gamma ^i_{jk} = (\Gamma ^Q)^i_{jk} =0$ {\em globally}.

With these coordinates, we have
\[
0= (Ric)_{uk} = \frac{ \partial (d^Q\Omega)_{ik}}{\partial x^i},
\]
or
\[
\frac{\partial}{\partial x^i} \left(\frac{\partial\Omega _i}{\partial x^k} - \frac{\partial \Omega _k}{\partial x^i} \right) =0.\, (k=1,2)
\]
The latter equation implies that the quantity
\begin{equation}\label{eq1}
\alpha := \frac{1}{2}\left(\frac{\partial\Omega _1}{\partial y} - \frac{\partial \Omega _2}{\partial x} \right)
\end{equation}
only depends on the parameter $u$, i.e., $\alpha\equiv\alpha(u)$. We may therefore define a new 1-parameter family of {\em closed} (thus, {\em exact}) 1-forms $\tilde{\Omega}(u)$ on $Q$ by
\[
\begin{array}{cc}
\tilde{\Omega}:=(-\alpha y+\Omega_1)dx + (\alpha x+\Omega_2)dy.
\end{array}
\]
Therefore, there exists some function $f \in C^{\infty}(\mathbb{R}\times Q)$ such that $\tilde{\Omega}=df$, and so,
\[
\Omega_1(u,x,y)  = \frac{\partial f}{\partial x}(u,x,y) + \alpha(u)y,\quad \Omega_2(u,x,y)  = \frac{\partial f}{\partial y}(u,x,y) - \alpha(u)x.
\]
Consider now the change of variable $V=v+f(u,x,y)$ which transforms the metric (\ref{ref2}) into

\begin{equation}\label{ref2-1}
2du\left(dV + \check{H}(u,x,y)du +\alpha(u)\left(ydx-xdy\right)\right) + dx^2+dy^2,
\end{equation}

\noindent where

\begin{equation}\label{eq4}
\check{H}(u,x,y):= H(u,x,y) - \frac{\partial f}{\partial u}(u,x,y).
\end{equation}
Then, all we need to do is remove the term  $\alpha(u)\left(ydx-xdy\right)$ to obtain a pp-wave. In order to achieve this, let us consider the following coordinates

\begin{equation}\label{eq3}
\begin{array}{rcl}
X &=& \cos (\beta(u)) x + \sin (\beta(u)) y, \\
Y &=& - \sin (\beta (u)) x + \cos (\beta (u)) y,
\end{array}
\end{equation}
where $\beta(u)=\int_0^u \alpha(s)ds$, and observe that

\[
dX^2+dY^2=dx^2 + dy^2 + \alpha^2(u)\left(x^2+y^2\right)du^2+2\alpha(u)du\left(ydx -xdy\right).
\]

In conclusion, in the coordinates $\{U,V,X,Y\}$ (with $U:=u$) we have that the metric (\ref{ref2-1}) becomes

\begin{equation}\label{ref3}
g = 2 dU (dV + \tilde{H}(U,X,Y)dU) + dX^2 + dY^2,
\end{equation}
where

\begin{equation}\label{eq:aux4}
\tilde{H}(U,X,Y):= H(u,x,y) - \frac{\partial f}{\partial u}(u,x,y)- \frac{\alpha^2(u)}{2} (x^2 + y^2).
\end{equation}

The condition that $(M,g)$ is Ricci-flat translates, in terms of these coordinates, into
\begin{equation}\label{g}
\frac{\partial ^2 \tilde{H} }{\partial X^2} +\frac{\partial ^2 \tilde{H} }{\partial Y^2} = 0.
\end{equation}
Now, $U$ and $V$ clearly have range $\mathbb{R}$, and the form of the metric (\ref{ref3}) implies the 2-dimensional submanifolds $U,V =const.$ are totally geodesic, and hence (since $(M,g)$ is geodesically complete) are {\em isometric} copies of the Euclidean 2-dimensional spaces. Hence, the range of both $X$ and $Y$ is also $\mathbb{R}$, which concludes the proof.
\end{proof}

\begin{remark}\label{rem:aux} {\em  In the particular case when the standard Brinkmann spacetime $(M,g)$ is autonomous, we deduce that the function $\alpha$ defined on (\ref{eq1}) is actually constant. Hence, the new coordinates $X$ and $Y$ can be viewed as arising from $x,y$, for each $u$, via a rotation of angle $\alpha\cdot u$. So, even if we start from an autonomous Brinkmann spacetime, the pp-wave obtained after the changes of variables discussed above is not necessarily autonomous. Note, however, that in this case $\tilde{H}$ has a very concrete dependence on $U(:=u)$, given precisely by the variable change \eqref{eq3}, and we deduce from \eqref{eq:aux4} that
	\begin{equation}\label{eq:aux5}
	\tilde{H}(U,X,Y)=\hat{H}(x,y):=H(x,y)-\frac{\alpha^2}{2}(x^2+y^2).
	\end{equation}

On the other hand, if we also assume that ${\rm span}\{{\partial_v,\partial_u}\}^{\perp}$ is integrable, then $\Omega$ is closed, and thus exact (as we are in the universal cover). Therefore, the function $\alpha$ not only is constant, but actually equal to zero, and the change of coordinates in the spatial fiber (\ref{eq3}) becomes trivial. We conclude that in this case, the arguments in Theorem \ref{thm:teo1} lead in fact to an autonomous pp-wave.}
	 \end{remark}

\section{Proof of Theorem \ref{main}}\label{sec3}

In order to prove Theorem \ref{main}, we will need some preliminary lemmas and definitions. The following definition was introduced (in slightly different form)
in Refs. \cite{FS,FlS}.
\begin{definition}
	\label{atmostquadratic} A function $\fun:\mathbb{R}^n \rightarrow \mathbb{R}$ is {\em at most quadratic} if there exist numbers $a,b>0$ such that
	\[
	\fun(x) \leq a\|x\|^ 2 + b, \forall x \in \mathbb{R}^n.
	\]
	
\end{definition}

\begin{remark}
	\label{r1}
	{\em Note that if a function $\fun: \mathbb{R}^n \rightarrow \mathbb{R}$ is {\em not} at most quadratic, then there exists a sequence $\{x_k\}_k$ in $\mathbb{R}^n$ for which
		\[
		\fun(x_k) > k \|x_k\|^ 2 + k, \forall k \in \mathbb{N},
		\]
and, in particular, $\|x_k\| \rightarrow +\infty$ as $k \rightarrow +\infty$.
Clearly, if $\fun$ remains bounded above by a polynomial of degree at most 2 outside a compact subset of $\mathbb{R}^n$ then $\fun$ is at most quadratic.}
\end{remark}

The importance of Definition \ref{atmostquadratic} in our context arises from the following result, due to H.P. Boas and R.P. Boas (see \cite{BB}, Thm. II).
\begin{lemma}
	\label{boasyboas}
	A harmonic function $\fun: \mathbb{R}^n \rightarrow \mathbb{R}$ bounded from one side by a polynomial of degree $m$ is also a polynomial of degree at most $m$. In particular, if $\fun$ is at most quadratic, then there exist numbers $a_{ij},b_j \in \mathbb{R}$ ($i,j \in \{ 1, \ldots, n\}$) such that
	\[
	\fun(x) = \sum_{i,j =1}^ n a_{ij}x^ix^j + \sum_{j=1}^ n b_j x^j + \fun(0).
	\]
\end{lemma}
\qcd

The following two technical lemmas will also be instrumental in our proof.
\begin{lemma}
	\label{jose1}
	Let $\Omega \subseteq \mathbb{C} \equiv \mathbb{R}^2$ be an open set containing $0$, and let $\fun: \Omega \rightarrow \mathbb{R}$ be a harmonic function such that $\fun(0) =0$. Then, for each $R>0$ such that $\overline{B_R(0)} \subset \Omega$, there exists a number $\theta _R\in [0,2\pi)$ for which
	\[
	\int_0^ R \fun(re^{i\theta _R})\,dr =0.
	\]
\end{lemma}
{\em Proof.} Fix one such $R>0$. Consider the continuous function $I_R: \theta \in [0, 2\pi) \mapsto I_R(\theta) \in \mathbb{R}$ given by
\[
I_R(\theta):= \int_0^ R \fun(re^{i\theta})\,dr,\;\; \forall \theta \in [0, 2\pi).
\]
Integrating this function on the interval $[0, 2\pi)$, we get
\[\begin{array}{rl}
\int_0^ {2\pi} I_R(\theta)\,d\theta =& \int_0^ {2\pi}\int_0^ R \fun(re^{i\theta})\,drd\theta= \int_0^ R \left(\int_0^ {2\pi} \fun(re^{i\theta})\,d\theta\right)\,dr\\  =& \int_0^ R 2 \pi r \fun(0)\,dr \equiv 0,
\end{array}\]
where we have used the mean value theorem for harmonic functions in the third equality. Hence, for some $\theta_R \in [0, 2 \pi)$ , $I_R(\theta_R) =0$ as claimed.
\qcd

\begin{lemma}
	\label{jose2}
	Let $\Omega \subseteq \mathbb{C} \equiv \mathbb{R}^2$ be an open set containing $0$, and let $\fun: \Omega \rightarrow \mathbb{R}$ be a harmonic function such that $\fun(0) =0$. Then, for each $R>0$ such that $\overline{B_R(0)} \subset \Omega$, and for each $p \in \partial B_R(0)$, there exists a piecewise smooth curve $z: [0,1] \rightarrow \overline{B_R(0)}$ such that
	\begin{itemize}
		\item[i)] $z(0) = z(1) =0$ and $z(t_0) =p$ for some $t_0 \in (0,1)$,
		\item[ii)] $\int_0^ 1 \fun(z(t)) dt \geq \frac{1}{5} \fun(p)$, and
		\item[iii)] $\int_0^ 1 \| \dot{z}(t)\|^ 2 dt \leq 50\pi^2 R^2$.
	\end{itemize}
\end{lemma}
{\em Proof.} Fix one such $R>0$, and let $\theta_R \in [0, 2\pi)$ be as in Lemma \ref{jose1}. Write $p = Re^{i\theta _0} \equiv (R \cos \theta _0, R \sin \theta _0)$. We may assume $\theta_0 \geq \theta_R$, since the case when $\theta_0 \leq \theta_R$ is entirely analogous.

Assume first that $\theta_0 = \theta_R$. In this case, we
define
\begin{equation}\label{curve1}
z(t)=\left \{\begin{array}{ll} \frac{5}{2}tR e^{i\theta_R} &
\quad\;\;\hbox{if}\;\; t\in [0,2/5] \\ p  &
\quad\;\;\hbox{if}\;\; t\in [2/5,3/5] \\  \frac{5}{2}(1-t)R e^{i\theta_R}
& \quad\;\;\hbox{if}\;\; t\in
[3/5,1].\end{array}\right.
\end{equation}
With this definition, appropriate changes of variables immediately show that
\[
\int_0^{2/5} \fun(z(t))dt = \int_{3/5}^{1} \fun(z(t))dt = (2/5R) \int_0^ R \fun(r e^{i\theta_R})\,dr = 0
\]
from the choice of $\theta_R$, and hence
\[
\int_0^{1} \fun(z(t))dt = \frac{1}{5}\fun(p).
\]
Moreover,
\[
\int_0^ 1 \| \dot{z}(t)\|^ 2 dt = 5R^2 < 50 \pi^ 2 R^2.
\]

We now assume that $\theta_0 > \theta_R$. Consider the standard parametrization $\gamma: t \in [0,2 \pi] \mapsto R e^ {i t} \in \mathbb{C}$ of the circle of radius $R$. By the mean value theorem for the harmonic function $\fun$ we have
\[
0= \fun(0) = \frac{1}{2\pi} \int_{\theta_R}^ {\theta_R+2\pi} \fun(\gamma(t)) \|\dot{\gamma} (t) \| dt,
\]
whence we conclude that
\begin{equation}
\label{sign}
0= \int_{\theta_R}^{\theta_R+2\pi} \fun(\gamma(t))  dt = \int_{\theta_R}^ {\theta_0} \fun(\gamma(t))  dt + \int_{\theta_0}^{\theta_R+2\pi} \fun(\gamma(t))  dt.
\end{equation}
We may consider two cases:
\[
(a) \int_{\theta_R}^ {\theta _0} \fun(\gamma(t))  dt\geq 0,\qquad \hbox{or}\qquad
(b)\int_{\theta_0}^{\theta_R+2\pi} \fun(\gamma(t))  dt \geq 0.
\]

For (a), consider the reparametrization $\beta: t \in [1/5,2/5] \mapsto \gamma(5(\theta_0 - \theta_R)t + 2 \theta_R - \theta_0)$ of the curve $\gamma$. Then
\[
0 \leq \int_{\theta_R}^ {\theta _0} \fun(\gamma(t))  dt = \frac{1}{5(\theta_0 - \theta_R)} \int_{1/5}^{2/5} \fun(\beta(s))ds,
\]
and since we assume $\theta_0 > \theta_R$ we conclude that
\begin{equation}
\label{needed}
\int_{1/5}^{2/5} \fun(\beta(s))ds \geq 0.
\end{equation}
Also, note that $\beta(1/5) = R e^ {i\theta_R}$ and $\beta(2/5) = R e^ {i\theta_0}=p$, and
\[
\int_{1/5}^{2/5} \| \dot{\beta}(t)\|^ 2 dt = 5  (\theta_0 - \theta_R)^2 R^2 \leq 20 \pi^ 2 R^2.
\]

We may therefore define $z:[0,1] \rightarrow \mathbb{C}$ by
\begin{equation}\label{curve2}
z(t)=\left \{\begin{array}{ll} 5tRe^{i\theta_R} &
\quad\;\;\hbox{if}\;\; t\in [0,1/5] \\
 \beta(t)  &
\quad\;\;\hbox{if}\;\; t\in [1/5,2/5] \\ p &\quad\;\;\hbox{if}\;\; t\in [2/5,3/5]\\ \beta(1-t)
& \quad\;\;\hbox{if}\;\; t\in
[3/5,4/5] \\ 5(1-t)Re^{i\theta_R}
& \quad\;\;\hbox{if}\;\; t\in
[4/5,1].\end{array}\right.
\end{equation}
Thus,
\begin{eqnarray}
\label{terms}
\int_{0}^ {1} \fun(z(t))  dt &=& \int_{0}^ {1/5} \fun(5tRe^{i\theta_R})  dt + \int_{1/5}^ {2/5} \fun(\beta(t))  dt + \frac{1}{5}\fun(p) \\ \nonumber
&+& \int_{3/5}^ {4/5} \fun(\beta(1-t))  dt+ \int_{4/5}^ {1} \fun(5(1-t)Re^{i\theta_R})  dt \\ \nonumber
&=& \frac{2}{5R}\int_0^ R \fun(r e^ {i\theta_R}) dr (\equiv 0)+ 2 \int_{1/5}^ {2/5} \fun(\beta(t))  dt (\geq 0)+ \frac{1}{5}\fun(p) \\ \nonumber
&\geq & \frac{1}{5}\fun(p),
\end{eqnarray}
while
\[
\int_{0}^{1} \| \dot{z}(t)\|^ 2 dt = 2 \int_{1/5}^{2/5} \| \dot{\beta}(t)\|^ 2 dt + 10 R^2  \leq 50 \pi^2 R^2,
\]
which concludes case (a). The case (b) follows analogously, just interchanging $\beta$ with a map $\tilde{\beta}: [1/5,2/5] \rightarrow \mathbb{C}$ defined by
\[
\tilde{\beta} (t) = \gamma(5(\theta_R+2\pi-\theta_0)t+2\theta_0-\theta_R-2\pi)
\]
(compare with the definition of $\beta$), and the result follows.

\qcd

{\em Proof of Theorem \ref{main}}.

Note that since $(M,g)$ is transversally Killing, Corollary \ref{thm:teoIJ2} ensures that the universal covering of $(M,g)$ is a standard autonomous Brinkmann spacetime. We can then assume without loss of generality that $(M,g)$ is a standard autonomous Brinkmann spacetime, and so, that $M=\mathbb{R}^2\times Q$ and $g$ is expressed as \eqref{SBM} where $H$, $\Omega$ and $\gamma$ do not depend on the variable $u$.

Observe that now we can apply Theorem \ref{thm:teo1}, which ensures the existence of coordinates $\{U,V,X,Y\}$ with $(U,V,X,Y)\in \mathbb{R}^4$ for which $g$ has the expression
\[
g=2dU(dV+\tilde{H}(U,X,Y)dU)+dX^2+dY^2,\quad\hbox{with $\tilde{H}$ harmonic in $X,Y$.}
\]
Moreover, due the fact that $(M,g)$ is autonomous, Remark \ref{rem:aux} ensures that $\tilde{H}(U,X,Y)=\hat{H}(x,y)$ (see \eqref{eq:aux5}).

We wish to show that $\tilde{H}$ is quadratic in the coordinates $X$, $Y$. Now, since the coordinate transformations for $X$ and $Y$ are {\em linear} in $x,y$ (cf. \eqref{eq3}), in order to accomplish this it is enough to show that $\hat{H}$ is a quadratic function of $x$ and $y$.

Assume then, by way of contradiction, that $\hat{H}$ is {\em not} quadratic as a function of $x,y$. Since $-\hat{H}$ is harmonic in $x,y$, due to Lemma \ref{boasyboas} $-\hat{H}$ {\em can not} be at most quadratic in these coordinates. Therefore (cf. Remark \ref{r1}) we can pick a sequence $p_k=(x_k,y_k)$ in $\mathbb{R}^2$ for which
\begin{equation}\label{eq:aux3}
-\hat{H}(p_k) > k \|p_k\|^ 2 + k,\quad \forall k \in \mathbb{N}
\end{equation}
and $R_k:= \|p_k\| \rightarrow +\infty$ as $k \rightarrow +\infty$.

Our strategy from now on is as follows. We will show the existence of some open set ${\cal U}_0$ containing the origin $(0,0,0,0)$ of $M\equiv \mathbb{R}^4$  and timelike curve segments with endpoints arbitrarily close to the origin, such that they are not contained in ${\cal U}_0$, in violation of our assumption of strong causality for $(M,g)$. This contradiction then yields that $\hat{H}$ is indeed quadratic, which in turn establishes the theorem.

 So, let us fix ${\cal U}_0$ be the open Euclidean ball in $\mathbb{R}^4$ centered at the origin and with radius $R_0>0$. Fix a number $0< \Delta <R_0$. We can assume that $R_k >R_0$ for all $k \in \mathbb{N}$, and so, that any point $(u,v,p_k)\notin {\cal U}_0$. For each $k \in \mathbb{N}$, we may use Lemma \ref{jose2} with $\fun=-\hat{H}$ and pick a piecewise smooth curve $z_k: t \in [0,1] \mapsto \overline{B_{R_k}(0)}\subset\mathbb{C}\equiv\mathbb{R}^2$ such that
\begin{itemize}
	\item[(i)] $z_k(0) = z_k(1) =(0,0)$ and $z(t_k) =p_k$ for some $t_k \in (0,1)$,
	\item[(ii)] -$\int_0^ 1 \hat{H}(z_k(t)) dt \geq -\frac{1}{5} \hat{H}(p_k)$, and
	\item[(iii)] $\int_0^ 1 \| \dot{z_k}(t)\|^ 2 dt \leq 50\pi^2 R_k^2$.
\end{itemize}
Using these curves, we may define for each $k \in \mathbb{N}$, the piecewise curve $Z_k : t \in [0,1] \mapsto \mathbb{C}$ given by $Z_k(t)=e^{i\alpha\Delta t}z_k(t)$
for each $t \in [0,1]$, where $\alpha$ is defined in (\ref{eq1}) and it is constant due the autonomous character of $(M,g)$. Observe that, from construction, $\tilde{H}(\alpha\Delta t,Z_k(t))=\hat{H}(z_k(t))$. Therefore, if we write $z_k(t) = x_k(t) + i y_k(t)$ we compute:
\begin{eqnarray}
\label{estimatederivative}
\|\dot{Z}_k\|^2=\dot{Z}_k\,\dot{\overline{Z}}_k &=&   \alpha^2 \Delta ^2 \|z_k\|^2 + \| \dot{z}_k\|^2 + i \alpha \Delta (z_k\dot{\overline{z}}_k - \dot{z}_k \overline{z}_k) \\ \nonumber
&=& \alpha^2 \Delta^2 R_k^ 2 +\|\dot{z}_k\|^2+2\alpha \Delta(x_k \dot{y}_k - y_k\dot{x}_k) \\ \nonumber
&\leq& \alpha^2 \Delta^2 R_k^ 2 +\|\dot{z}_k\|^2+4|\alpha| \Delta R_k \|\dot{z}_k\| \\ \nonumber
&\leq& 3\alpha^2 \Delta^2 R_k^ 2 +3\|\dot{z}_k\|^2.
\end{eqnarray}
Joining the previous inequality with (iii) we conclude that
\begin{equation}
\label{estimatederivative2}
\int_0^ 1 \| \dot{Z_k}(t)\|^ 2 dt \leq C(\alpha, \Delta) R_k^ 2,
\end{equation}
where
\[
C(\alpha,\Delta) = 3\alpha^2\Delta^2 + 150\pi^2.
\]
Finally, for each number $E>0$ and for each $k \in \mathbb{N}$, we can define the curve $\Gamma_k^E:[0,1] \rightarrow \mathbb{R}^4$ given, for each $t \in [0,1]$, by
\[
\Gamma_k^E(t) := (V_k^ E(t), \Delta t, Z_k(t)),
\]
where
\begin{eqnarray}
\label{keydef}
V_k^E(t)&:=& -\Delta \int _0^t \tilde{H}(\alpha \Delta s, Z_k(s))ds - \frac{1}{2\Delta}\int_0 ^t \|\dot{Z}_k(s) \|^2 ds - \frac{Et}{2\Delta} \\ \nonumber
&=& - \Delta \int _0^t \hat{H}(z_k(s))ds - \frac{1}{2\Delta}\int_0 ^t \|\dot{Z}_k(s) \|^2 ds - \frac{Et}{2\Delta}.
\end{eqnarray}
It is easy to check, using the line element of $g$ in the form (\ref{ref3}), that each $\Gamma_k^E$ defines a timelike curve in $(M,g)$.

Now, consider the smooth functions $h_k: E \in (0,+\infty) \mapsto V_k^E(1) \in \mathbb{R}$ ($k \in \mathbb{N}$). Clearly, for each $k \in \mathbb{N}$, $h_k(E) < 0$ for large enough $E$. However, collecting our estimates (ii), \eqref{eq:aux3} and \eqref{estimatederivative2}, we get from (\ref{keydef})
\begin{equation}
\label{climax}
h_k(1) \geq (\frac{k}{5} \Delta  - \frac{C(\alpha, \Delta)}{2\Delta})R^2_k + \frac{k}{5} - \frac{1}{2\Delta}.
\end{equation}
It is clear from inequality (\ref{climax}) that we can pick $k_0 \in \mathbb{N}$ for which $h_{k_0}(1) >0$, and since $h_{k_0}$ is continuous, there exists $E_0 >0$ for which $h_{k_0}(E_0) =0$. We then conclude that $\Gamma^{E_0}_{k_0}$ is a timelike curve such that $\Gamma^{E_0}_{k_0}(0) = (0,0,0,0)$ and $\Gamma^{E_0}_{k_0}(1) = (0, \Delta, 0,0) \in {\cal U}_0$ but $\Gamma^{E_0}_{k_0}(t_{k_0})=(V_{k_0}^{E_0}(t_{k_0}),\Delta t_{k_0},p_k) \notin {\cal U}_0$, as desired; so the proof is complete.
\qcd

In the Introduction, we have shown how Theorem 1.2 implies, via Corollary \ref{keycor}, that a physically relevant but relatively restricted class of pp-waves actually fall under the important Cahen-Wallach subclass. Our final goal in this paper is to widen the scope of that result so as to encompass the larger class of those Brinkmann spacetimes envisaged in Theorem \ref{main}, and give a precise, concrete geometric characterization of when these are Cahen-Wallach spaces.

In order to achieve this, let $(M,g)$ be a Brinkmann spacetime in the conditions of Theorem \ref{main}. Corollary \ref{thm:teoIJ2} allows us to assume without loss of generality that $(M,g)$ is a standard autonomous Brinkmann spacetime. Now we assume, in addition, that ${\rm span}\{\paral,\killi\}^{\perp}$ is an integrable distribution, being $\paral$, $\killi$ the corresponding parallel and Killing null fields. Then, by applying Theorem \ref{thm:teo1}, and taking into account Remark \ref{rem:aux}, we deduce that $(M,g)$ is isometric to an autonomous pp-wave, i.e. with $H$ independent of $u$. So, if we take up again the arguments of the proof of Theorem \ref{main}, we conclude that $H$ must be quadratic in the spatial coordinates, and thus a Cahen-Wallach space. Summarizing:
\begin{corollary}
Let $(M,g)$ be a Brinkmann spacetime in the conditions of Theorem \ref{main} and denote by $\paral$ and $\killi$, respectively, the corresponding parallel and Killing null fields. Then, the universal cover of $(M,g)$ is a Cahen-Wallach spacetime if and only if ${\rm span}\{\paral,\killi\}^{\perp}$ is an integrable distribution.
\end{corollary}

\section*{Acknowledgments}

The authors are partially supported by the Spanish Grant MTM2013-47828-C2-2-P (MINECO and FEDER funds). JLF is also supported by the Regional J. Andaluc\'{i}a grant PP09-FQM-4496, with FEDER funds. IPCS wishes to acknowledge a research scholarship from CNPq, Brazil, Programa Ci\^{e}ncia sem Fronteiras, process number 200428/2015-2. JH would like to thank the Department of Algebra, Geometry and Topology of the University of M\'{a}laga, JLF the Department of Mathematics of the Universidade Federal de Santa Catarina, while IPCS extends his warm thanks to the faculty and staff members of the Department of Mathematics of the University of Miami, for their kind hospitality while part of the work on this paper was being carried out.

\end{document}